\documentclass[11pt,english]{article}
\pdfoutput=1
\usepackage{lmodern}

\usepackage[T1]{fontenc}
\usepackage[latin9]{inputenc}
\usepackage{geometry}
\usepackage{babel}
\usepackage{array}
\usepackage{units}
\usepackage{amsthm}
\usepackage{amsmath}
\usepackage{amssymb}
\usepackage[unicode=true,pdfusetitle,
 bookmarks=true,bookmarksnumbered=false,bookmarksopen=false,
 breaklinks=false,pdfborder={0 0 1},backref=false,colorlinks=false]
 {hyperref}

\makeatletter

\newcommand{\noun}[1]{\textsc{#1}}
\providecommand{\tabularnewline}{\\}

\theoremstyle{plain}
\newtheorem{thm}{\protect\theoremname}
  \theoremstyle{definition}
  \newtheorem{defn}[thm]{\protect\definitionname}
  \theoremstyle{plain}
  \newtheorem*{lem*}{\protect\lemmaname}
  \theoremstyle{plain}
  \newtheorem*{thm*}{\protect\theoremname}
  \theoremstyle{remark}
  \newtheorem*{claim*}{\protect\claimname}

\usepackage[protrusion=true,expansion=true]{microtype}
\usepackage{tikz}
\usepackage{fancyhdr}

\makeatother

  \providecommand{\claimname}{Claim}
  \providecommand{\definitionname}{Definition}
  \providecommand{\lemmaname}{Lemma}
  \providecommand{\theoremname}{Theorem}
\providecommand{\theoremname}{Theorem}

\begin{document}

\title{Component mixers and a hardness result for counterfeiting quantum
money}

\author{Andrew Lutomirski}

\date{Friday, July 1, 2011}

\maketitle
\global\long\def\op#1{\mathop\textrm{#1}\nolimits}

\renewcommand{\headrulewidth}{0pt}\thispagestyle{fancyplain}\rhead{MIT-CTP 4279}
\begin{abstract}
In this paper we give the first proof that, under reasonable assumptions,
a problem related to counterfeiting quantum money from knots \cite{luto-knot-money}
is hard. Along the way, we introduce the concept of a component mixer,
define three new classical query problems and associated complexity
classes related to graph isomorphism and group membership, and conjecture
an oracle separating QCMA from QMA.
\end{abstract}

\section{Introduction}

Quantum money from knots is a cryptographic protocol in which a mint
produces a quantum state $|\$_{\ell}\rangle$. Anyone can verify that
the state came from the mint, and hopefully it is intractable to copy
the state. The quantum state is described by a set $S$ which is partitioned
into components. (In practice $S$ is a set of knot diagrams with
bounded complexity and the components are the sets of knots with the
same Alexander polynomial.)

We hope to prove the security of an abstracted version of the protocol
in which adversaries have only black-box access to an idealized version
of knot-theoretic operations \cite{luto-qmoney-ics}. In this abstracted
version, the quantum state is a superposition of $n$-bit strings
($n$ is a security parameter chosen by the mint). All of those strings
come from a large set $S$, and the state $|\$_{\ell}\rangle=\sum_{x\in S_{\ell}}|x\rangle$
is the uniform superposition of strings in the $\ell^{\text{th}}$
component.

All parties (the mint, honest users of money, and any adversaries)
have access to two black-box operations. They have access to a {}``component
mixer'' (defined below) that invertibly maps any string to a new
almost uniformly random string in the same component. They also have
access to a labeling function that determines which component any
element is in. (In the concrete scheme, the component mixer does not
fully mix within the components. We ignore that issue here.)

To prove hardness results related to counterfeiting quantum money,
we need an appropriate computational assumption. We find this assumption
in a new class of query problems based on component mixers.

All of these problems involve a large set that is partitioned into
components. An algorithm must use black-box queries to a component
mixer to answer questions about the components. The algorithm is not
given access to a labeling function. The \noun{same component} problem
is: are two given elements in the same component? The \noun{multiple
components} problem is: is there more than one component (as opposed
to just one)? If we promise that either there is only one component
or no component contains a majority of the set, then the \noun{multiple
components} problem becomes \noun{multiple balanced components}. Finally,
on a quantum computer, the \noun{component superposition} problem
is to prepare the uniform superposition of all elements in a component
given as input one element in that component. (The classical analog,
producing a uniformly random sample from a component, is easy by assumption.)

These types of questions are natural abstractions of graph isomorphism
and group membership. Graph isomorphism is the problem of deciding
whether two graphs are equivalent up to a permutation. The complexity
of graph isomorphism is unknown on both classical and quantum computers.
Group membership, on the other hand, is the problem of deciding whether
an element of some large group is in a particular subgroup of that
group. The subgroup is specified by its generators, and the group
structure is given either as a black box or in some explicit form
such as a matrix representation. In the black-box setting, group membership
is hard on classical computers \cite{Watrous2000Groups} but has unknown
complexity on quantum computers.

For graph isomorphism, the big set would be the set of all graphs
of a given size and the components are isomorphism classes of graphs.
The component mixer permutes the vertices of a graph. Testing whether
two graphs are isomorphic reduces to the \noun{same component} problem.

For group membership, the big set would be a large group and the components
would be cosets of a subgroup that is described only by its generators.
The component mixer would multiply by an element of the subgroup.
The group membership problem reduces to an instance of \noun{same
component}: testing whether a given element is in the same component
as the identity. The \noun{multiple balanced components} problem would
determine whether the given generators generate the entire group or
a proper subgroup.

Each of these query problems naturally defines a complexity class.
SCP, MCP, and MBCP are the sets of languages that are polynomial-time
reducible to the \noun{same component} problem, \noun{multiple components}
problem and \noun{multiple balanced components} problem. We relate
all three classes to commonly-used complexity classes. Our results
are summarized in the table above.

These problems and classes are immediately interesting for two reasons.
First, if \noun{same component} is hard on a quantum computer, then
we have evidence for the security of a quantum money protocol \cite{luto-knot-money}.
Second, MCP and MBCP are candidates for a classical oracle separation
between QCMA and QMA. (Group membership does not work directly because
it has too much structure \cite{AK07}.)

\begin{table}
\begin{centering}
\begin{tabular}{c>{\centering}p{5cm}>{\centering}p{4.8cm}}
The class\ldots{} & \ldots{}is in\ldots{} & \ldots{}and oracle-separated from\tabularnewline
\hline 
SCP & NP, SZK & co-MA, hopefully BQP\tabularnewline
MCP & QMA, $\text{NP}^{\text{co-NP}}$, $\text{NP}^{\text{co-SCP}}$ & BQP, hopefully QCMA\tabularnewline
MBCP & MCP, AM, co-AM, SZK, $\text{BPP}^{\text{SCP}}$ & hopefully QCMA\tabularnewline
\end{tabular}
\par\end{centering}

\caption{\label{tab:complexity-zoo}Component problems have a home in the complexity
zoo.}
\end{table}

\section{Definitions}

Throughout this paper, we use some basic terminology. We say that
$a\in_{R}B$ if $a$ is a uniform random sample from $B$. A function
$f:\mathbb{N}\to\mathbb{R}$ is negligible if for all $y$ there exists
$N_{y}$ such that $f\left(x\right)<x^{-y}$ for all $x>N_{y}$. Intuitively,
negligible functions go to zero faster than the reciprocal of any
polynomial. Finally, the total variation distance between two distributions
with probability density functions $p$ and $q$ over a set $D$ is
\[
\frac{1}{2}\sum_{e}\left|p(e)-q(e)\right|=\sup_{A\subseteq D}\left|p(A)-q(A)\right|.
\]
 The total variation distance is sometimes referred to as the statistical
difference, and it is analogous to the trace distance between density
matrices of mixed quantum states.

All of the problems we consider are questions about a large set $S$.
For consistency in defining the size of the problems, we take $n$
to be the number of bits used to represent an element of $S$. The
set $S$ is partitioned into components, and access to the components
is given through a family of invertible maps that takes an element
of any component of $S$ to a new element of the same component. The
maps constitute a component mixer if a uniform random choice of the
map produces a uniformly random output.
\begin{defn}
A family of one-to-one maps $\left\{ M_{i}\right\} $ is a \emph{component
mixer} on a partition $\left\{ S_{1},\ldots,S_{c}\right\} $ of a
set $S$ if:\end{defn}
\begin{itemize}
\item The set $S$ is a subset of $n$-bit strings.
\item The family is indexed by a label $i$ from a set $\op{Ind}_{M}$,
and each $i$ can be encoded in $O\left(\op{poly}\left(n\right)\right)$
bits.
\item The functions $\left\{ M_{i}\right\} $ do not mix between components.
That is, for all $i$ and $a$, if $x\in S_{a}$ then $M_{i}\left(x\right)\in S_{a}$
as well.
\item The functions $\left\{ M_{i}\right\} $ instantly mix within each
component. That is, for all $a$ and $x\in S_{a}$, if $i\in_{R}\op{Ind}_{M}$,
then the total variation distance between $M_{i}\left(x\right)$ and
a uniform sample from $S_{a}$ is no more than $2^{-n-2}$.
\end{itemize}

The last condition is often easy to satisfy directly. For graph isomorphism,
if $\op{Ind}_{M}$ is the set of permutations, $M_{i}$ can apply
the permutation $i$ to a graph. For group membership, if $\op{Ind}_{M}$
is a sequence of coin flips that can generate a nearly uniform sample
from the subgroup (e.g. a straight-line program as in \cite{Babai:1991:LEV:103418.103440}),
then $M_{i}$ can multiply by the element of the subgroup implied
by the coin flips. In general, given a Markov chain over $S$ that
does not mix between components but mixes rapidly over each component,
each step of which consists of choosing uniform random sample from
a set of invertible rules and applying that rule, then iterating that
Markov chain to amplify its spectral gap will give a component mixer.

In graph isomorphism, generating a random graph, testing whether some
encoding of a graph is valid, and generating a random permutation
to apply to the vertices are all easy. Similarly, in group membership,
generating a random element of the whole group (as opposed to the
subgroup) is easy, as is generating a random element of the subgroup.
When we abstract these problems to component mixer problems, we want
the corresponding operations to be easy as well. This leads to our
definition of query access to a component mixer.
\begin{defn}
An algorithm has query access to a component mixer $\left\{ M_{i}\right\} $
if the algorithm can do each of the following operations in one query
with failure probability no more than $2^{-n}$.\end{defn}
\begin{itemize}
\item Test an $n$-bit string for membership in $S$.
\item Generate an uniform random sample from $S$.
\item Test a string for membership in $\op{Ind}_{M}$.
\item Generate an uniform random sample from $\op{Ind}_{M}$.
\item Given $s\in S$ and $i\in\op{Ind}_{M}$, compute $M_{i}\left(s\right)$.
\item Given $s\in S$ and $i\in\op{Ind}_{M}$, compute $M_{i}^{-1}\left(s\right)$.
\end{itemize}

If we are considering \emph{quantum} algorithms, we want to give the
algorithm some quantum power. For example, in graph isomorphism, generating
a uniform quantum superposition of \emph{all} graphs is easy, as is
generating a uniform quantum superposition of all members of the permutation
group \cite{grover-2002}. We give quantum component algorithms the
equivalent powers.
\begin{defn}
An algorithm has quantum query access to a component mixer $\left\{ M_{i}\right\} $
if the algorithm can do each of the following operations coherently
in one query with failure probability no more than $2^{-n}$:\end{defn}
\begin{itemize}
\item Test an $n$-bit string for membership in $S$.
\item Generate the state $\sum_{s\in S}|s\rangle$ or measure the projector
onto that state.
\item Test a string for membership in $\op{Ind}_{M}$.
\item Generate the state $\sum_{i\in\op{Ind}_{M}}|i\rangle$ or measure
the projector onto that state.
\item Compute the {}``controlled-$M$'' operator, abbreviated $\op{CM}$.
$\op{CM}$ takes three registers as input: the first is the number
$-1$, $0$, or $+1$, the second is a string $i$, and the third
is an $n$-bit string $s$. On input $|\alpha,i,s\rangle$, $\op{CM}|\alpha,i,s\rangle=|\alpha,i,M_{i}^{\alpha}\left(s\right)\rangle$
if $i\in\op{Ind}_{M}$ and $s\in S$; otherwise $\op{CM}|\alpha,i,s\rangle=|\alpha,i,s\rangle$.
\end{itemize}
As a technical detail, we assume that any algorithm given (quantum)
query access to a component mixer $\left\{ M_{i}\right\} $ knows
both $n$ and the number of bits needed to encode an element of $\op{Ind}_{M}$.

We can now state the definitions of our query problems.
\begin{defn}
The \noun{same component} problem is: given query access to a component
mixer $\left\{ M_{i}\right\} $ on a set $S$ and two elements $\left(s,t\right)\in S$,
accept if $s$ and $t$ are in the same component of $S$.
\end{defn}

\begin{defn}
The \noun{multiple components} problem is: given query access to a
component mixer $\left\{ M_{i}\right\} $ on a partition $\left\{ S_{1},\ldots,S_{c}\right\} $,
accept if $c>1$.
\end{defn}

\begin{defn}
The \noun{multiple balanced components} problem is: There is a partition
$\left\{ S_{1},\ldots,S_{c}\right\} $ with the promise that either
there is only one component or no component contains more than half
the elements in $S$. Given query access to a component mixer $\left\{ M_{i}\right\} $
on that partition and the string $0^{n}$, accept if $c>1$.
\end{defn}

On a quantum computer, we can also try to generate the uniform superposition
over a component.
\begin{defn}
The \noun{component superposition} problem is: given quantum query
access to a component mixer $\left\{ M_{i}\right\} $ on a set $S$
and an element $s\in S$, output the state 
\[
|S_{j}\rangle=\frac{1}{\sqrt{\left|S_{j}\right|}}\sum_{u\in S_{j}}|u\rangle
\]
 where $S_{j}$ is the (unknown) component containing $s$.
\end{defn}

The decision problems can also be viewed as complexity classes. We
define the class SCP to be the set of languages that are polynomial-time
reducible to the \noun{same component }problem with bounded error.
Similarly, we define MCP by reference to \noun{multiple components}
and MBCP by reference to \noun{multiple balanced components}.

\section{\label{sec:basic-props}Basic properties of component mixers}
\begin{lem*}
\emph{(Component mixers are fully connected)} If $s$ and $t$ are
in the same component, then there exists $i$ such that $t=M_{i}\left(s\right)$.\end{lem*}
\begin{proof}
Assume the contrary. Suppose $s$ and $t$ are in the same component
$S_{j}$ and let $A=\left\{ M_{i}\left(s\right):i\in\op{Ind}_{M}\right\} $.
By assumption, $t\notin A$. This means that the variation distance
between $M_{i}\left(s\right)$ (for $i\in_{R}\op{Ind}_{M}$) and a
uniform sample on $S_{j}$ is 
\begin{align*}
 & \frac{1}{2}\sum_{u\in S_{j}}\left|\op{Pr}\left[M_{i}\left(s\right)=u\right]-\frac{1}{\left|S_{j}\right|}\right|\\
\ge & \frac{1}{2}\left|\op{Pr}\left[M_{i}\left(s\right)=t\right]-\frac{1}{\left|S_{j}\right|}\right|\\
= & \frac{1}{2\left|S_{j}\right|}\\
\ge & 2^{-n-1}>2^{-n-2},
\end{align*}
 which contradicts the fact that $M$ is a component mixer.
\end{proof}
A uniform quantum superposition over all of the elements in one component
is a potentially useful state. It is not obvious whether a quantum
computer can produce or verify such a state with a small number of
queries to a component mixer, but it is possible to verify that a
state is in the span of such superpositions.
\begin{lem*}
\emph{(Quantum computers can project onto component superpositions)
}A quantum computer can, with a constant number of queries to a component
mixer, measure the projector 
\[
P=\sum_{k}\left(\left|S_{k}\right|^{-\nicefrac{1}{2}}\sum_{x\in S_{k}}|x\rangle\right)\left(\left|S_{k}\right|^{-\nicefrac{1}{2}}\sum_{x\in S_{k}}\langle x|\right)=\sum_{k}\left(\frac{1}{\left|S_{k}\right|}\sum_{x,y\in S_{k}}|x\rangle\langle y|\right)
\]
 with negligible error as a function of $n$.\end{lem*}
\begin{proof}
Starting with a state $|\psi\rangle$, we give an algorithm to measure
$P$ on $|\psi\rangle$ . The algorithm uses three registers: $|\psi\rangle$
starts in register $A$; registers $B$ and $C$ are ancillas. $B$'s
computational basis is $\op{Ind}_{M}$ and $C$ holds a single bit.
To simplify the notation, we write the uniform superposition in register
$B$ as $|e_{0}\rangle=\left|\op{Ind}_{M}\right|^{-\nicefrac{1}{2}}\sum_{i}|i\rangle_{M}$.
The algorithm is:
\begin{enumerate}
\item Initialize register $B$ to $|e_{0}\rangle_{B}$ and register $C$
to $|0\rangle$. This gives the state 
\[
|\phi_{1}\rangle=|\psi\rangle_{A}|e_{0}\rangle_{B}|0\rangle_{C}.
\]

\item Apply controlled-$M$ with the control set to 1. This is equivalent
to applying $M$ unconditionally. Let $\tilde{M}_{j}$ be the quantum
operator corresponding to the action of $M_{j}$ on register $A$.
That is, $\langle s'|\tilde{M}_{j}|s\rangle=\langle s'|M_{j}(s)\rangle$.
With this notation, the action of this step on registers $A$ and
$B$ is $U=\left(\sum_{j}\tilde{M}_{j}\otimes|j\rangle\langle j|\right)$.
The resulting state is 
\begin{align*}
|\phi_{2}\rangle & =U|\psi\rangle_{A}|e_{0}\rangle_{B}|0\rangle_{C}\\
 & =|e_{0}\rangle_{BB}\langle e_{0}|U|\psi\rangle_{A}|e_{0}\rangle_{B}|0\rangle_{C}+\left(1-|e_{0}\rangle_{BB}\langle e_{0}|\right)U|\psi\rangle_{A}|e_{0}\rangle_{B}|0\rangle_{C}
\end{align*}

\item Apply the unitary operator $|e_{0}\rangle_{BB}\langle e_{0}|\otimes X_{C}+\left(1-|e_{0}\rangle_{BB}\langle e_{0}|\right)\otimes I_{C}$.
This sets register $C$ to $|1\rangle$ is register $B$ is still
in the state $|e_{0}\rangle$. The state is now 
\[
|\phi_{3}\rangle=|e_{0}\rangle_{BB}\langle e_{0}|U|\psi\rangle_{A}|e_{0}\rangle_{B}|1\rangle_{C}+\left(1-|e_{0}\rangle_{BB}\langle e_{0}|\right)U|\psi\rangle_{A}|e_{0}\rangle_{B}|0\rangle_{C}.
\]

\item Uncompute step 2 by applying $U^{\dagger}$. This gives 
\begin{align*}
|\phi_{4}\rangle & =U^{\dagger}|e_{0}\rangle_{BB}\langle e_{0}|U|\psi\rangle_{A}|e_{0}\rangle_{B}|1\rangle_{C}+U^{\dagger}\left(1-|e_{0}\rangle_{BB}\langle e_{0}|\right)U|\psi\rangle_{A}|e_{0}\rangle_{B}|0\rangle_{C}\\
 & =\left(U^{\dagger}|e_{0}\rangle_{BB}\langle e_{0}|U\right)|\psi\rangle_{A}|e_{0}\rangle_{B}|1\rangle_{C}+\left(1-U^{\dagger}|e_{0}\rangle_{BB}\langle e_{0}|U\right)|\psi\rangle_{A}|e_{0}\rangle_{B}|0\rangle_{C}.
\end{align*}
 To simplify this result, observe that the matrix $\left|\op{Ind}_{M}\right|^{-1}\sum_{j}\tilde{M}_{j}$
is the Markov matrix obtained by applying one of the $M_{i}$ uniformly
at random to an element of $S$. From the definition of a component
mixer, $\left|\op{Ind}_{M}\right|^{-1}\sum_{j}\tilde{M}_{j}\approx P$
. Furthermore, $P|\psi\rangle$ has the form $\alpha\sum_{x\in S_{a}}|x\rangle$
for some $a$ and $\alpha$, and $M_{k}$ preserves the set $S_{a}$,
so $\tilde{M}_{k}P|\psi\rangle=P|\psi\rangle$ for all $k$. Using
these observations, we can simplify 
\begin{align*}
U^{\dagger}|e_{0}\rangle_{BB}\langle e_{0}|U|\psi\rangle_{A}|e_{0}\rangle_{B} & =\left(\sum_{k}\tilde{M}_{k}^{\dagger}\otimes|k\rangle_{BB}\langle k|\right)|e_{0}\rangle_{BB}\langle e_{0}|\left(\sum_{j}\tilde{M}_{j}\otimes|j\rangle\langle j|\right)|\psi\rangle_{A}|e_{0}\rangle_{B}\\
 & =\left(\sum_{k}\tilde{M}_{k}^{\dagger}\otimes|k\rangle_{BB}\langle k|\right)|e_{0}\rangle_{B}\left|\op{Ind}_{M}\right|^{-1}\left(\sum_{j}\tilde{M}_{j}\right)|\psi\rangle_{A}\\
 & \approx\left(\sum_{k}\tilde{M}_{k}^{\dagger}\otimes|k\rangle_{BB}\langle k|\right)P|\psi\rangle_{A}|e_{0}\rangle_{B}\\
 & =P\left(\sum_{k}|k\rangle_{BB}\langle k|\right)|\psi\rangle_{A}|e_{0}\rangle_{B}\\
 & =P|\psi\rangle_{A}|e_{0}\rangle_{B}.
\end{align*}
 Plugging this in, we have 
\[
|\phi_{4}\rangle\approx P|\psi\rangle_{A}|e_{0}\rangle_{B}|1\rangle_{C}+\left(1-P\right)|\psi\rangle_{A}|e_{0}\rangle_{B}|0\rangle_{C}
\]
 with negligible error.
\end{enumerate}
At this point, register $B$ is unentangled with the rest of the system,
register $C$ contains the outcome of the measurement we wanted, and
register $A$ contains the correct final state.
\end{proof}
On a quantum computer, \noun{same component} reduces to \noun{component
superposition}: given two initial elements, a swap test can decide
with bounded error whether their respective component superpositions
are the same state or non-overlapping states.

\section{Placing component mixer problems in the complexity zoo}

\subsection{Inclusions}

Several of the complexity class relationships in Table~\ref{tab:complexity-zoo}
are straightforward. \noun{Multiple components} is a relaxation of
\noun{multiple balanced components}, so $\op{MBCP}\subseteq\op{MCP}$.
The {}``component mixers are fully connected''  lemma implies a
simple NP algorithm for \noun{same component}, so $\op{SCP}\subseteq\op{NP}$.
\noun{Multiple components} can be restated as {}``do there exist
two objects that are \emph{not} in the same component?'', so $\op{MCP}\subseteq\op{NP}^{\op{co-SCP}}$
and hence $\op{MCP}\subseteq\op{NP}^{\op{co-NP}}$. 

In the appendix, we give two Arthur-Merlin protocols for \noun{multiple
balanced components}:
\begin{itemize}
\item A protocol to prove a {}``yes'' answer. In this protocol, Merlin
solves the \noun{same component} problem on input given by Arthur
(appendix \ref{MBCP_AM}).
\item A protocol to prove a {}``no'' answer (appendix \ref{MBCP_co-AM}).
\end{itemize}
The existence of these protocols implies that $\op{MBCP}\subseteq\op{AM},\op{BPP}^{\op{SCP}},\text{ and }\op{co-AM}$.
We also give a QMA protocol for \noun{multiple components} (see appendix~\ref{MCP_QMA}).

\noun{Same component} is reducible to \noun{statistical difference}:
to test whether $s$ and $t$ are in the same component, choose $i,j\in_{R}\op{Ind}_{M}$
and test whether $M_{i}(s)$ and $M_{j}(t)$ have the same distribution.
\noun{Statistical difference} is complete for SZK, so $\op{SCP}\subseteq\op{SZK}$
\cite{SahaiVadhanSZK}.

\noun{Multiple balanced components} also reduces to \noun{statistical
difference}: choose $a,b\in_{R}S$ and $i,j\in_{R}\op{Ind}_{M}$.
If there are multiple balanced components, then the predicate that
the first two and last two elements of $\left(a,M_{i}(a),b,M_{j}(b)\right)$
are in the same component holds w.p.\ 1, whereas the same predicate
holds on four independent uniform samples from $S$ w.p.\ at most
$\nicefrac{1}{4}$. This means that the variation distance between
$\left(a,M_{i}(a),b,M_{j}(b)\right)$ and four independent samples
is at least $\nicefrac{3}{4}$. If, on the other hand, there is only
one component, then $\left(a,M_{i}(a),b,M_{j}(b)\right)$ is negligibly
different four independent samples from $S$. Therefore, \noun{multiple
balanced components} reduces to \noun{statistical difference} on the
distribution of $\left(a,M_{i}(a),b,M_{j}(b)\right)$ versus four
independent uniform samples from $S$. Hence $\op{MBCP}\subseteq\op{SZK}$.

\subsection{Separations}

SCP contains group membership (relative to any oracle) and group membership
is not in co-MA for black-box groups \cite{Watrous2000Groups}, so
$\op{SCP}\nsubseteq\text{co-MA}$ relative to an oracle.

The quantum query complexity of \noun{multiple components} is exponential
by reduction from the Grover problem (see appendix~\ref{sec:MCP_not_BQP}).
This implies the existence of an oracle separating MCP and BQP.

\subsection{Conjectured separations}

We conjecture that there is no QCMA or co-QCMA proof for \noun{multiple
components} or even \noun{multiple balanced components}, which would
imply the existence of an oracle separating MBCP from QMA and hence
QCMA from QMA.

We further conjecture that \noun{multiple balanced components} has
superpolynomial randomized and quantum query complexity. This conjecture
would imply that MBCP is separated from BPP and BQP by an oracle.

\section{A hardness result for counterfeiting quantum money}

We are now ready to prove a hardness result for counterfeiting quantum
money. Recall that the quantum money state is defined \cite{luto-qmoney-ics,luto-knot-money}
as 
\[
|\$_{\ell}\rangle=\sum_{x\in S_{\ell}}|x\rangle
\]
 where $S_{\ell}$ is a component of a partition of a big set $S$
and an adversary is given access to a component mixer for that partition.
Unlike the other component mixer problems we have discussed, an adversary
also has access to a labeling function $L$ that maps each element
of $S$ to a label that identifies which component that element is
in.

We show that, if an attacker is given one copy of $|\$_{\ell}\rangle$
and \emph{measures} it in the computational basis, then, under reasonable
assumptions, the attacker cannot recreate the state. That is, given
some $s\in S_{\ell}$ (i.e. the measurement outcome), it is hard to
produce $|\$_{\ell}\rangle$. We call this type of attack \noun{simple
counterfeiting}. Our assumption is that the quantum query complexity
of \noun{same component} is superpolynomial.
\begin{defn}
The \noun{simple counterfeiting} problem is: given quantum query access
to a component mixer $\left\{ M_{i}\right\} $ on a set $S$, quantum
query access to a function $L$ that maps each element of $S$ to
a unique label identifying the component containing that element,
and an element $s\in S$, output the state 
\[
|S_{j}\rangle=\frac{1}{\sqrt{\left|S_{j}\right|}}\sum_{u\in S_{j}}|u\rangle
\]
 where $S_{j}$ is the component containing $s$. 
\end{defn}
\noun{Simple counterfeiting} is the same problem as \noun{component
superposition} except that the algorithm also has access to the labeling
function. This makes the problem seem easier; for example, \noun{same
component} and \noun{multiple balanced components} both become trivial
with access to the labeling function. We show that the labeling function
is unhelpful for the purpose of \noun{simple counterfeiting}.
\begin{thm*}
If the quantum query complexity of \noun{component superposition}
is superpolynomial, then the quantum query complexity of \noun{simple
counterfeiting} is also superpolynomial.\end{thm*}
\begin{proof}
The \noun{simple counterfeiting} and \noun{component superposition}
problems differ in that \noun{simple counterfeiting} is given access
to a label that identifies components. Calculating such a label given
only a component mixer is at least as hard as solving \noun{simple
counterfeiting} in the first place, so we won't be able to provide
a valid label. The idea behind the proof is to show that a correct
labeling function is not very helpful for solving \noun{simple counterfeiting},
and that, given a component mixer, we can efficiently provide a label
that is indistinguishable from a valid label in polynomial time.

We assume for contradiction that we have a quantum query algorithm
{}``alg'' that solves \noun{simple counterfeiting} in $n^{k}$ queries
for sufficiently large $n$. Alg is given quantum query access to
a component mixer and labeling function and it is promised that the
labeling function is consistent with the component mixer. It takes
as input an element $s\in S_{j}$ for some $j$. It makes $n^{k}$
quantum queries and produces a mixed state $\rho$ as output. The
trace distance between $\rho$ and the desired output state $\frac{1}{\sqrt{\left|S_{j}\right|}}\sum_{u\in S_{j}}|u\rangle$
is a negligible function of $n$.

We give an algorithm that solves \noun{component superposition} with
high probability using alg as a subroutine.

As input, we have quantum query access to a component mixer on $n$
bits and an $n$-bit string $s$. This means that the space of $n$
bit strings is partitioned into components $S_{1},\ldots,S_{c}$ and
a set of {}``garbage'' strings $G=\left\{ 0,1\right\} ^{n}\setminus\left(S_{1}\cup\cdots\cup S_{c}\right)$,
where $c$ is the (unknown) number of components. We are not given
access to a labeling function. WLOG, we assume that $s\in S_{1}$.

We define an instance of \noun{simple counterfeiting} on $2n$-bit
strings that can be used to solve the original \noun{component superposition}
problem. To simplify the notation, we treat each $2n$-bit string
as a pair of binary numbers, each between 0 and $2^{n}-1$. In our
instance of \noun{simple counterfeiting}, the components are $\left\{ 0\right\} \times S_{1},\ldots,\left\{ 0\right\} \times S_{c}$
and $\left\{ 0\right\} \times G$. Each other element (that is, everything
that has something nonzero as its first $n$ bits) is its own component.
We use the component mixer 
\[
M_{i}^{(0)}\left(r,z\right)=\begin{cases}
\left(0,M_{i}\left(z\right)\right) & \text{ if }r=0\\
\left(r,z\right) & \text{ otherwise}
\end{cases}
\]
 and \emph{incorrect} label 
\[
L^{(0)}\left(r,z\right)=\begin{cases}
\left(0,0\right) & \text{ if }r=0\\
\left(r,z\right) & \text{ otherwise}
\end{cases}.
\]

The label $L^{\left(0\right)}$ violates the promise of \noun{simple
counterfeiting} (it assigns the same label to all of the components
in the original component mixer), so the \noun{simple counterfeiting}
algorithm run directly on $\left\{ M_{i}^{\left(0\right)}\right\} $
and $L^{(0)}$ might fail. However, the only way to detect that $L^{(0)}$
is invalid is to query it on some input of the form $\left(0,t\right)$
for $t\in S_{2}\cup\cdots\cup S_{c}$. Those inputs are an exponentially
small fraction of the domain of $L^{(0)}$ and we can hide them by
randomly permuting $L^{(0)}$ and $M_{i}^{(0)}$, giving this algorithm: 
\begin{enumerate}
\item Choose independent random permutations $\pi$ and $\sigma$ on $\mathbb{Z}_{2^{n}}\times\mathbb{Z}_{2^{n}}$.
$\pi$ indicates where each $2n$-bit string is hidden in the permuted
problem and $\sigma$ scrambles the labels. (These permutations will
take an exponential number of bits to specify, but they can be implemented
with no queries to $\left\{ M_{i}\right\} $.)
\item Run alg on $\left\{ \pi\circ M_{i}^{(0)}\circ\pi^{-1}\right\} $ and
$\sigma\circ L^{(0)}\circ\pi^{-1}$ with the initial element $\pi\left(0,s\right)$.
\item Apply $\pi^{-1}$ coherently to the quantum state that alg produces.
\item Output the last $n$ qubits of the result.
\end{enumerate}
If $\sigma\circ L^{(0)}\circ\pi$ were a valid label function for
the component mixer $\left\{ \pi\circ M_{i}^{(0)}\circ\pi\right\} $,
then this algorithm would succeed on each try w.p.\ negligibly different
from 1. We will prove that the invalidity of the labeling function
is well enough hidden that the algorithm works anyway.

To prove this, we assume the contrary: there is some $\left\{ M_{i}\right\} $
for which this algorithm fails with nonnegligible probability. This
means that the actual output of our algorithm differs nonnegligibly
in trace distance from the desired output. Such a difference would
be detectable if we knew what the correct output was; we will show
that this is impossible by solving the Grover problem more quickly
than is allowed by the BBBV theorem using alg as a subroutine.

We generalize the functions $M_{i}^{(0)}$ and $L^{(0)}$ to a larger
family that encodes a Grover search problem. We can picture $\{M_{i}^{(0)}\}$
as an embedding of the original problem in the first row of a grid
in which the first $n$ bits is the row index and the last $n$ bits
is the column index (see Figure~\ref{fig:embeddings}---the unmarked
squares are their own components). There are many other ways we could
have embedded the original problem, though. (These other embeddings
are well-defined, but they are difficult to calculate without access
to a labeling function for the original problem.) In particular, we
could have placed everything except $S_{1}$ on a different row. If
we put the other components on the $j^{\text{th}}$ row, we get 
\[
L^{(j)}\left(r,z\right)=\begin{cases}
\left(0,0\right) & \text{ if }r=0\text{ and }z\in S_{1}\\
\left(0,0\right) & \text{ if }r=j\text{ and }z\notin S_{1}\\
\left(r,z\right) & \text{ otherwise}
\end{cases}
\]
 and 
\[
M_{i}^{(j)}\left(r,z\right)=\begin{cases}
\left(0,M_{i}\left(z\right)\right) & \text{ if }r=0\text{ and }z\in S_{1}\\
\left(j,M_{i}\left(z\right)\right) & \text{ if }r=j\text{ and }z\notin S_{1}\\
\left(r,z\right) & \text{ otherwise}
\end{cases}.
\]
 Alternatively, we could leave them out entirely, giving

\[
L^{\text{nowhere}}\left(r,z\right)=\begin{cases}
(0,0) & \text{ if }r=0\text{ and }z\in S_{1}\\
(r,z) & \text{ otherwise}
\end{cases}
\]
 and 
\[
M_{i}^{\text{nowhere}}\left(r,z\right)=\begin{cases}
\left(0,M_{i}\left(z\right)\right) & \text{ if }r=0\text{ and }z\in S_{1}\\
\left(r,z\right) & \text{ otherwise}
\end{cases}.
\]

We can't efficiently implement queries to $L^{\text{nowhere}}$, $M^{\text{nowhere}}$,
$L^{(j)}$ or $M_{i}^{(j)}$ for $j\ne0$, but, if we could and if
\noun{simple counterfeiting} didn't notice that the label function
was invalid, then the output on any of instances with starting element
$\left(0,s\right)$ would be 
\[
\sum_{z\in S_{1}}|0\rangle|z\rangle,
\]
 the latter $n$ qubits of which is exactly the state we wanted.

The function $L^{\text{nowhere}}$ is a valid labeling function, but
all of the $L^{(j)}$ are invalid because they take the same value
on the images of $S_{1},\ldots,S_{c}$ even though they are in different
components. Nonetheless, they look valid as long as no one ever queries
them on the images of $S_{2},\ldots,S_{c}$, which collectively represent
less than a $2^{-n}$ fraction of all possible queries.

We formalize this notion by a reduction from the Grover problem. Suppose
$g:\mathbb{Z}_{2^{n}}\to\left\{ 0,1\right\} $ is a function that
outputs 1 at most one input. By the BBBV theorem \cite{Bennett:1997:SWQ:264393.264407},
the query complexity of distinguishing a random point function $g$
from all zeros is $O\left(2^{n/2}\right)$. Using our algorithm for
\noun{simple counterfeiting} as a subroutine, we will attempt to decide
whether $g$ maps any value to 1. We do this by allowing $g$ to select
which embedding to use. This gives the {}``labeling'' function 
\[
L^{\left[g\right]}\left(r,z\right)=\begin{cases}
(0,0) & \text{ if }r=0\text{ and }z\in S_{1}\\
(0,0) & \text{ if }g(r)=1\text{ and }z\notin S_{1}\\
(r,z) & \text{ otherwise}
\end{cases}
\]
 and component mixer 
\[
M_{i}^{\left[g\right]}\left((r,z)\right)=\begin{cases}
\left(0,M_{i}\left(z\right)\right) & \text{ if }r=0\text{ and }z\in S_{1}\\
\left(j,M_{i}\left(z\right)\right) & \text{ if }g(r)=1\text{ and }z\notin S_{1}\\
\left(r,z\right) & \text{ otherwise}
\end{cases}.
\]
 If $g(j)=1$ for some $j$, then $L^{\left[g\right]}=L^{\left(j\right)}$
and $M_{i}^{\left[g\right]}=M_{i}^{\left(j\right)}$; otherwise $L^{\left[g\right]}=L^{\text{nowhere}}$
and $M^{\left[g\right]}=M^{\text{nowhere}}$. It is possible to evaluate
either $L^{\left[g\right]}\left(r,z\right)$ or $M_{i}^{\left[g\right]}\left(r,z\right)$
with a very large number of queries to the original component mixer
$\left\{ M_{i}\right\} $ and \emph{one} query to $g\left(r\right)$.
(Evaluating the functions coherently requires a second query to $g\left(r\right)$
to uncompute garbage.)

If we choose independent random permutations $\pi$ and $\sigma$
on $\mathbb{Z}_{2^{n}}\times\mathbb{Z}_{2^{n}}$ and run alg on $\left\{ \pi\circ M_{i}^{[g]}\circ\pi{}^{-1}\right\} $
and $\sigma\circ L^{[g]}\circ\pi-1$ with initial state $\pi\left(0,s\right)$,
the output of the algorithm is some mixed state that depends on $g$.
Let $\rho_{0}$ be the density matrix of that mixed state if $g$
is all zeros and let $\rho_{\text{point}}$ be the density matrix
if $g$ is a uniformly random point function.
\begin{claim*}
$\left\Vert \rho_{0}-\rho_{\text{point}}\right\Vert _{\text{tr}}$
is a negligible function of $n$.\end{claim*}
\begin{proof}
Assume the contrary: $\left\Vert \rho_{0}-\rho_{\text{point}}\right\Vert _{\text{tr}}\ge n^{-k}$
for some fixed $k$ and an infinite sequence of values of $n$. If
we run alg on $L^{\left[g\right]}$ and $M_{i}^{\left[g\right]}$,
we can then decide whether the output is $\rho_{0}$ or $\rho_{\text{point}}$
and therefore whether $g$ is all zeros or a point function by measuring
the output. We will get the right answer w.p.\ at least $\frac{1}{2}+\frac{n^{-k}}{2}$.
We can amplify $n^{2k}$ times to get the right answer w.p.\ at least
$\nicefrac{2}{3}$ by a Chernoff bound. By assumption, alg makes $n^{r}$
queries to $L^{\left[g\right]}$ and $M_{i}^{\left[g\right]}$. That
means that, in $n^{r+2k}=o\left(n^{n/2}\right)$ queries, we can determine
whether $g\left(j\right)=1$ for any $j$, which is impossible by
the BBBV theorem. Therefore $\left\Vert \rho_{0}-\rho_{\text{point}}\right\Vert _{\text{tr}}$
is a negligible function of $n$.
\end{proof}
It follows that, if we apply $\pi{}^{-1}$ to $\rho_{0}$ and to $\rho_{\text{point}}$,
the results differ negligibly in trace distance. The result of applying
$\pi^{-1}$ to $\rho_{0}$ is the uniform superposition over $\left\{ 0\right\} \times S_{1}$
up to negligible error because if $g=0$ then alg's promise is satisfied
and it produces the correct answer. Furthermore, if we set $g\left(0\right)=1$,
then the output distribution is still $\rho_{0}$ because the distribution
of component mixers and labels seen by alg is independent of which
point function we choose. This means that $\pi{}^{-1}$ applied to
the output of alg on $\left\{ \pi\circ M_{i}^{(0)}\circ\pi^{-1}\right\} $
and $\sigma\circ L^{(0)}\circ\pi^{-1}$ with initial state $\pi\left(0,s\right)$
differs negligibly from the uniform superposition over $\left\{ 0\right\} \times S_{1}$
in trace distance.

This contradicts the assumption that there exists some input on which
our algorithm fails, so our algorithm solves \noun{component superposition}
with negligible error.
\end{proof}
We can replace the assumption that \noun{component superposition}
is hard with the assumption that \noun{same component} is hard because
\noun{same component} reduces to \noun{component superposition}.

\begin{figure}
\begin{centering}
\def\textat(#1,#2,#3){\path[anchor=base] node at (#1+0.5,#2+0.3) {#3}}
\begin{tikzpicture}[x=0.5cm,y=0.5cm]
 \path[anchor=east] node at (-0.1,0.6) {$\vdots$};
 \path[anchor=east] node at (0,1.5) {3};
 \path[anchor=east] node at (0,2.5) {2};
 \path[anchor=east] node at (0,3.5) {1};
 \path[anchor=east] node at (0,4.5) {0};
 \path[anchor=south] node at (7,6) {Last $n$ bits (not in order)};
 \draw[step=1] (0,0) grid (14,4);
 \draw (0,4) -- (0,5) -- (14,5) -- (14,4);
 \draw (3,4) -- (3,5) (7,4) -- (7,5) (9,4) -- (9,5) (12,4) -- (12,5);
 \textat(1,4,$S_1$);
 \textat(4.5,4,$S_2$);
 \textat(7.5,4,$S_3$);
 \textat(10,4,$\cdots$);
 \textat(12.5,4,$G$);
 \draw[->] (16.5,2.5) -- (15,2.5);
 \path[anchor=west] node at (17,2.6) {Components in $\{M_i^{(0)}\}$};

\begin{scope}[shift={(0,-6.5)}]
 \path[anchor=east] node at (-0.1,0.6) {$\vdots$};
 \path[anchor=east] node at (0,1.5) {3};
 \path[anchor=east] node at (0,2.5) {2};
 \path[anchor=east] node at (0,3.5) {1};
 \path[anchor=east] node at (0,4.5) {0};
 \path[anchor=south] node[rotate=90] at (-1.8,2.5) {First $n$ bits};
 \draw[step=1] (0,0) grid (14,3);
 \draw (3,4) -- (0,4) -- (0,5) -- (3,5);
 \draw[step=1] (3,4) grid (14,5);
 \draw (0,4) -- (0,3) (3,4) -- (3,3) (7,4) -- (7,3) (9,4) -- (9,3) (12,4) -- (12,3) (14,4) -- (14,3);
 \draw (1,4) -- (1,3) (2,4) -- (2,3);
 \textat(1,4,$S_1$);
 \textat(4.5,3,$S_2$);
 \textat(7.5,3,$S_3$);
 \textat(10,3,$\cdots$);
 \textat(12.5,3,$G$);
 \draw[->] (16.5,2.5) -- (15,2.5);
 \path[anchor=west] node at (17,2.6) {Components in $\{M_i^{(1)}\}$};
\end{scope}

\begin{scope}[shift={(0,-13)}]
 \path[anchor=east] node at (-0.1,0.6) {$\vdots$};
 \path[anchor=east] node at (0,1.5) {3};
 \path[anchor=east] node at (0,2.5) {2};
 \path[anchor=east] node at (0,3.5) {1};
 \path[anchor=east] node at (0,4.5) {0};
 \draw[step=1] (0,0) grid (14,4);
 \draw (0,4) -- (0,5) -- (14,5) -- (14,4);
 \draw \foreach \x in {3,4,5,6,7,8,9,10,11,12,13} {(\x,4) -- (\x,5) };
 \textat(1,4,$S_1$);
 \draw[->] (16.5,2.5) -- (15,2.5);
 \path[anchor=west] node at (17,2.6) {Components in $\{M_i^\text{nowhere}\}$};
\end{scope}
\end{tikzpicture}
\par\end{centering}

\caption{\label{fig:embeddings}There are $2^{n}+1$ ways to hide components
that are hard to label.}
\end{figure}
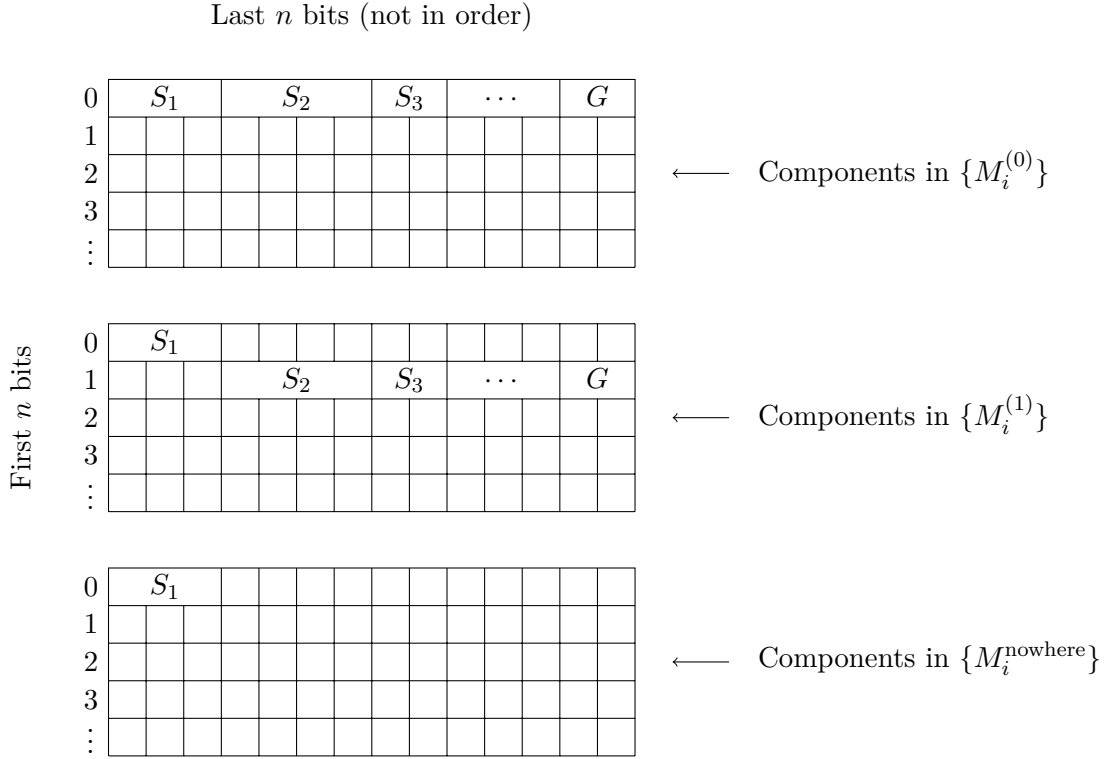

\section{Open problems}

There are a number of open problems related to this work.

Ideally, we would prove the impossibility of more general forms of
counterfeiting. If we could show that, given one copy of $|\$_{\ell}\rangle$
for some $\ell$, it is hard to produce a second copy of $|\$_{\ell}\rangle$,
then we would know that (in a black-box model) quantum money could
not be counterfeited. An even better result would be collision-freedom:
that is hard for anyone to produce a state of the form $|\$_{\ell}\rangle\otimes|\$_{\ell}\rangle$
by any means, even for a random $\ell$ of an attacker's choice. (Collision-freedom
implies that copying is impossible: if an attacker could copy a given
quantum money state, then the output of the algorithm would be contain
two copies of $|\$_{\ell}\rangle$ for the value of $\ell$ implied
by the input.)

It should be possible to prove quantum lower bounds on the query complexity
of \noun{same component} and \noun{multiple balanced components}.
This would strengthen the hardness result for counterfeiting quantum
money.

A classical oracle separating MCP and QCMA would also separate QMA
and QCMA. We conjecture that an appropriate worst-case component mixer
would work, but we have no proof.

A cryptographically secure component mixer could be a useful object,
and a good cryptographically secure component mixer with an associated
labeling function would give a better quantum money protocol than
quantum money from knots. (Knot invariants have all kinds of unnecessary
properties.) If we had that as well as a hardness result for generating
quantum money collisions, then quantum money would be on a sound theoretical
footing.

\section{Acknowledgments}

I would like to thank Eddie Farhi for extensive comments, Scott Aaronson
for suggestions about complexity classes, Peter Shor and Avinatan
Hassidim for encouraging me to look for reductions from problems like
graph isomorphism to counterfeiting quantum money, and Jon Kelner
for valuable thoughts about component mixers.

This work was supported by the Department of Defense (DoD) through
the National Defense Science \& Engineering Graduate Fellowship (NDSEG)
Program as well as the U.S. Department of Energy under cooperative
research agreement No. DE-FG02-94ER40818.

\bibliographystyle{utphys}
\phantomsection\addcontentsline{toc}{section}{\refname}\bibliography{same_component}

\appendix

\section{Query protocols for component problems}

\subsection{\label{MBCP_AM}An AM query protocol \noun{multiple balanced components}}

Suppose that Merlin wants to prove to Arthur that some component mixer
has multiple balanced components. Arthur and Merlin run this protocol:

\medskip{}

\begin{center}
\noindent %
\begin{tabular}{>{\raggedright}p{0.3cm}>{\raggedright}p{5.5cm}>{\raggedright}p{5.5cm}}
 & Arthur & Merlin\tabularnewline
\hline 
1. & Choose $s_{1},s_{2}\in_{R}S$, $i\in_{R}\left\{ 1,2\right\} $ and
$j\in_{R}\op{Ind}_{M}$. & \tabularnewline
2. & Compute $t=M_{j}\left(s_{i}\right)$. & \tabularnewline
3. & Send $s_{1},s_{2},t$ to Merlin. & \tabularnewline
4. &  & If $s_{1}$ and $s_{2}$ are in different components, compute $i'=i$.
Otherwise, choose $i'\in_{R}\left\{ 1,2\right\} $.\tabularnewline
5. &  & Send $i'$ to Arthur.\tabularnewline
6. & Accept iff $i=i'$. & \tabularnewline
\end{tabular}
\par\end{center}

\medskip{}

If $\left\{ M_{i}\right\} $ has multiple balanced components, then
with probability at least $\nicefrac{1}{2}$, $s_{1}$ and $s_{2}$
are in different components. In this case, Merlin will always answer
correctly. This means that Merlin is correct w.p.\ at least $\nicefrac{3}{4}$.
If, on the other hand, $M$ has only one component, then $t$ is a
nearly uniform sample from $S$ (trace distance at most $2^{-n-2}\le\nicefrac{1}{8}$).
This means that Merlin can guess $i$ correctly with probability at
most $\nicefrac{5}{8}$. With constant overhead, this protocol can
be amplified to give soundness and completeness errors $\nicefrac{1}{3}$.

Steps 1, 2, 3, 5, and 6 can be done in a constant number of queries
to the component mixer oracle. Step 4 requires Merlin to solve the
\noun{same component} to decide whether $t$ is in the same component
as $s_{1}$, $s_{2}$, or both. This means that if Arthur had the
power of SCP (with oracle access to $\left\{ M_{i}\right\} $), then
he could run the protocol on his own.

\subsection{\label{MBCP_co-AM}A co-AM query protocol for \noun{multiple balanced
components}}

Suppose that Merlin wants to prove to Arthur that some component mixer
has a single component (as opposed to multiple balanced components).
Arthur and Merlin run this protocol:

\medskip{}

\begin{center}
\noindent %
\begin{tabular}{>{\raggedright}p{0.3cm}>{\raggedright}p{5.5cm}>{\raggedright}p{5.5cm}}
 & Arthur & Merlin\tabularnewline
\hline 
1. & Choose $s_{1},s_{2}\in_{R}S$. & \tabularnewline
2. & Send $s_{1},s_{2}$ to Merlin. & \tabularnewline
3. &  & Choose $i\in_{R}\op{Ind}_{M}$ such that $M_{i}\left(s_{1}\right)=s_{2}$.\tabularnewline
4. &  & Send $i$ to Arthur.\tabularnewline
5. & Accept iff $M_{i}\left(s_{1}\right)=s_{2}$. & \tabularnewline
\end{tabular}
\par\end{center}

\medskip{}

If there is only one component, then $s_{1}$ and $s_{2}$ are in
the same component and Merlin can find $i$ because component mixers
are fully connected. If, on the other hand, there are multiple balanced
components, then w.p.\ at least $\nicefrac{1}{2}$, $s_{1}$ and
$s_{2}$ are in different components and no such $i$ exists.

This means that this proof is complete and has soundness error at
most $\nicefrac{1}{2}$. A constant amount of amplification will reduce
the soundness error below $\nicefrac{1}{3}$.

\subsection{\noun{\label{MCP_QMA}}A quantum witness for \noun{multiple components}}

Given a {}``yes'' instance of \noun{multiple components} problem,
let $S_{1}$ and $S_{2}$ be two distinct components. Then a valid
witness state is 
\[
|\psi_{\text{MC}}\rangle=\left(\sum_{s\in S_{1}}|s\rangle\right)\otimes\left(\sum_{s\in S_{2}}|s\rangle\right).
\]

To verify the witness, Arthur first measures the projector of each
register onto the space of uniform superpositions over components
(see section~\ref{sec:basic-props}). If either measurement outputs
zero, Arthur rejects. Otherwise Arthur performs a swap test between
the two registers and accepts iff the swap test says that the registers
are different.

On a valid witness, Arthur's projections succeed with probability
close to 1. The states in the two registers have disjoint support
(both before and after the swap test), so the swap test indicates
that the states are different w.p.\ $\nicefrac{1}{2}$. Arthur therefore
accepts a valid witness w.p.\ $\nicefrac{1}{2}$.

If there is only one component then projecting onto the space of uniform
superpositions over components is equivalent to projecting onto the
uniform superposition over $S$. Therefore, on any witness, if Arthur's
projections succeed then the post-measurement state is (up to negligible
error) two copies of the uniform superposition over $S$. Those two
copies are approximately the same state, so the swap test says that
they are the same and Arthur rejects w.p.\ near 1. Standard techniques
can amplify this protocol to give completeness and soundness errors
less than $\nicefrac{1}{3}$.

\section{\textmd{\noun{\normalsize \label{sec:MCP_not_BQP}Multiple components}}\textmd{\normalsize{}
has exponential quantum query complexity}}

We can embed an instance of the Grover problem into \noun{multiple
components}. Let $g$ be the instance of the Grover problem on $n$
bits (i.e. $g:\mathbb{Z}_{2^{n}}\to\left\{ 0,1\right\} $ is either
all zeros or a point function). Let $\op{Ind}_{M}=\mathbb{Z}_{2^{n}}$
and define the component mixer 
\[
M_{i}(x)=\begin{cases}
\left(x+i\right)\op{mod}2^{n} & \text{ if }g(x)=g(x+i)=0\\
x & \text{ otherwise}
\end{cases}.
\]
 If $g$ is all zeros then there is a single component but if $g(y)=1$
then $y$ is in its own component. The function $M_{i}$ can be evaluated
with two queries to $g$, so the Grover decision problem on $g$ reduces
to \noun{multiple components} on $\left\{ M_{i}\right\} $.

Hence, by the BBBV theorem \cite{Bennett:1997:SWQ:264393.264407},
the quantum query complexity of \noun{multiple components} is $\Omega\left(2^{n/2}\right)$.
\end{document}